\documentclass[final,authoryear,comma,5p,times,twocolumn]{elsarticle}
\usepackage{graphics}
\usepackage{amssymb}
\usepackage{amsthm}
\usepackage{amsfonts}
\usepackage{amsmath}
\usepackage{booktabs}
\usepackage{rotating}

\newtheorem{theorem}{Theorem}
\newtheorem{corollary}[theorem]{Corollary}

\newtheorem{proposition}[theorem]{Proposition}

\newdefinition{definition}{Definition}

\journal{Mathematical Social Sciences}

\begin{document}

\begin{frontmatter}

\title{Mathematical aspects of degressive proportionality}

\author[1]{Wojciech S\l omczy\'{n}ski\corref{cor1}}
\ead{wojciech.slomczynski@im.uj.edu.pl}
\cortext[cor1]{Corresponding author}
\author[2]{Karol \.{Z}yczkowski}
\address[1]{Institute of Mathematics, Faculty of Mathematics and Computer Science, Jagiellonian University, ul. {\L}ojasiewicza 6, 30-348 Krak\'{o}w, Poland}
\address[2]{Wojty{\l}a Institute, ul. Garncarska 5, 31-115 Krak\'{o}w, Poland}

\begin{abstract}
We analyze mathematical properties of apportionment functions in the context
of allocating seats in the European Parliament. Some exemplary
families of such functions are specified and the corresponding
allocations of seats among the Member States of the European Union are
presented. We show that the constitutional constraints for the apportionment
are so strong that the admissible functions lead to rather similar solutions.
\end{abstract}


\end{frontmatter}

\section{Introduction}
\label{INT}

One of the major mathematical approaches to the problem of allocating seats
in the European Parliament can be described by the following general scheme.
First, one has to choose a concrete characterization of the size of a given
Member State $i$ by a number $p_{i}$ (for example, equal to the total number
of its inhabitants, citizens, or voters\footnote{Of course, other more
exotic choices are also possible. According to the original text of the
Constitution of the United States (Article I, Section 2) `Representatives
(...) shall be apportioned among the several States which may be included
within this Union, according to their respective Numbers, which shall be
determined by adding to the whole Number of free Persons, (...) three fifths
of all other Persons.' The words `other Persons' here mean the slaves. The
rule was the result of the so called `Three-Fifths compromise' between
Southern and Northern states during the Constitutional Convention in 1787.})
we call here \textsl{population}, and precisely define by which means these
data should be collected and how often they should be updated. Then, one needs
to transform these numbers by an \textsl{allocation }(or
\textsl{apportionment}) \textsl{function} $A$ belonging to a given family
indexed (usually monotonically and continuously) by some \textsl{parameter
}$d$, whose range of variability is determined by the requirement that the
function fulfills constraints imposed by the treaties: is
\textsl{non-decreasing} and \textsl{degressively proportional}.

Additionally, the apportionment function satisfies certain boundary
conditions, $A\left(  p\right)  =m$ and $A\left(  P\right)  =M$, where the
population of the smallest and the largest state equals, respectively, $p$ and
$P$, and the smallest and the largest number of seats are predetermined as,
respectively, $m$ and $M$. (In the case of the European Parliament these
quantities are explicitly bounded by the treaty, $m\geq M_{\min}=6$ and $M\leq
M_{\max}=96$.) To obtain integer numbers of seats in the Parliament one has to
round the values of the allocation function, e.g., using one of three standard
\textsl{rounding methods} (upward, downward or to the nearest integer).
Finally, one has to choose the parameter $d$ in such a way that the sum of the
seat numbers of all Member States equals the given Parliament size $S$, solving
(if possible) in $d$ the equation%

\begin{equation}
\sum_{i=1}^{N}\left[  A_{d}\left(  p_{i}\right)  \right]  =S\text{ ,}
\label{main}%
\end{equation}
where $N$ stands for the number of Member States, $p_{i}$ for the population
of the $i$-th state ($i=1,\ldots,N$), and $\left[  \cdot\right]  $ denotes the
rounded number. Though usually there is a whole interval of parameters
satisfying this requirement, nonetheless, in a generic case, the distribution
of seats established in this way is unique. Thus, this technique bears a
resemblance to divisor methods in the proportional apportionment problem
applied first by Thomas Jefferson in 1792 \citep{BalYou78,Top09}.

The crucial role in this apportionment scheme plays the notion of
\textsl{degressive proportionality}. The principle of degressive
proportionality enshrined in the Lisbon Treaty was probably borrowed from the
discussions on the taxation rules, where the term has appeared already in the
nineteenth century, when many countries introduced income tax for the first
time in their history \citep{You94}. It was already included in the debate on
the apportionment in the Parliament in late 1980s, but at first, it was a
rather vague idea that gradually evolved into a formal legal (and
mathematical) term in the report \citet{LamSev07} adopted by the European
Parliament. There were also suggestions to apply this
general principle to other parliamentary or quasi-parliamentary bodies like
the projected Parliamentary Assembly of the United Nations \citep{Bum10}.

In fact the entire problem of apportionment of seats in the Parliament is
mathematically similar (not counting rounding) to the taxation problem, what
is illustrated in the table below.

\begin{table}[htbp]\footnotesize
\centering
\begin{tabular}{cc}
\toprule
\textbf{apportionment} & \textbf{taxation} \\
\midrule
Member States & Tax payers\\
Population & Income\\
Seats & Post-tax income\\
Allocation function & Post-tax income function\\
Parliament size & Total disposable income \\
Aeats monotonicity & Income order preservation\\
Degressivity of seats distribution & Progressivity of tax distribution\\
Subadditivity of seats distribution & Merging-proofness\\
\bottomrule
\end{tabular}%
\label{tab0}%
\end{table}%

In consequence, the similar mathematical tools can be used to solve both of
them; see for instance \citet{You87,Tho03,Kam06,Hou09,JuMor11}, and \citet{Mor11}, where the
authors use the above presented scheme to consider possible parametric
solutions of the taxation problem or the dual profit-sharing problem. Of
course, the analogy has clear limitations since income and post-tax income are
calculated in the same units, whereas population and seats are not. Moreover,
money is (at least theoretically) infinitely divisible, while seats are indivisible.

Although quite a novelty in politics, nevertheless, the concept of degressive
proportionality is not entirely new in mathematics. It was already analysed in
late 1940s under the name of `quasi-homogeneity' by \citet[Definition
1.4.1]{Ros50}, see also \citet[p.~480]{Kuc09}, and since then studied also under
the name of `subhomogeneity', see e.g. \citet{BurSza05}. Moreover, an
increasing function such that its inverse is degressively proportional (and so
it is an allocation function) is called `star-shaped' (with respect to the
origin) in the mathematical literature. In other words, the function is
degressively proportional if and only if the lines joining points lying below
its graph with the origin do not cross the graph. Star-shaped functions were
introduced in \citet{BruOst62}, and since then have
been studied in many areas of pure and applied mathematics, see e.g.
\citet{DinWol09,Dah10}. Thus, the results concerning this class of functions
can be applied, \textit{mutatis mutandis}, to degressively proportional functions.

Note that in the original definition of the degressive proportionality
formulated in \citet{LamSev07} it was postulated that this property holds for
the number of seats \textsl{after rounding} the values of the allocation
function to whole numbers. However, one can show that there exist such
distributions of population that there is no solution of the apportionment
problem satisfying so understood degressive proportionality
\citep{Ram10,Grietal11a}. In particular, such difficulty arises in situations
where there are a number of Member States having similar populations.
Consequently, in \citet{Grietal11a} it was recommended to weaken this condition
and to amend the definition of degressive proportionality assuming that `the
ratio between the population and the number of seats of each Member State
\textsl{before rounding} to whole numbers must vary in relation to their
respective populations in such a way that each Member from a more populous
Member State represents more citizens than each Member from a less populous
Member State'. This proposal has been recently approved by the The
Constitutional Affairs' Committee of the European Parliament (AFCO). For the
detailed mathematical analysis of the original definition of the degressive
proportionality, see \citet{Lyketal10,Ceg11,Flo11,Rametal11}, and \citet{Ser11}.

In this paper we describe several exemplary families of allocation functions
and discuss their fundamental properties. Mathematical technicalities
collected in Sect.~3-5 can be skipped by more practically oriented readers,
who may proceed to Sect.~\ref{EP}, in which general results are applied to
the European Parliament.

\section{Allocation functions - definition and examples\label{ALL1}}

Before selecting an allocation function $A$ one needs to specify the
boundary conditions $m$ and $M$, which denote the number of seats for the
smallest and the largest member state, with population $p$ and $P$,
respectively. In the case of the European Parliament, the treaty sets the
following bounds only: $m\geq M_{\min}=6$ and $M\leq M_{\mathrm{\max}}=96$.

\begin{definition}
Let $0<p<P$, $0<m<M$, and $pM<Pm$. We call $A:\left[  p,P\right]
\rightarrow\left[  m,M\right]  $ a \textsl{(degressive) allocation function}, if:

\begin{enumerate}
\item (monotonicity) $A$ is non-decreasing;

\item (degressive proportionality) $A$ is \textsl{degressively proportional},
i.e. the function $t\rightarrow A\left(  t\right)  /t$ is non-increasing.
\end{enumerate}

We shall also consider the situation where $P=M=+\infty$, assuming then that
$A:\left[  p,+\infty\right)  \rightarrow\left[  m,+\infty\right)  $. For the
sake of brevity we shall omit the word `degressive' and instead of saying that
`$A$ is a degressive allocation function' we shall simply say that `$A$ is an
allocation function'.
\end{definition}

Below, we consider several families of allocation functions fulfilling
additionally boundary conditions: $A\left(  p\right)  =m$ and $A\left(
P\right)  =M$. Each of them depends on one (free) parameter ($d$) with its
range of variability determined by other assumptions imposed on $A$. For
instance, in case of the allocation of seats in the Parliament, the parameter
$d$ is set by the constraint (\ref{main}) that the total size of the House is
fixed.

Note also that the actual value of the constant $d$ changes from one allocation
function to another.

\begin{enumerate}
\item \textsl{base+prop }functions - the `floor' version:
\begin{equation}
A_{1a}(t):=\max\left[  m,\;\left(  t-P\right)  /d+M\right]  \text{ ,}
\label{basproflo}%
\end{equation}

where $\frac{P}{M}\leq d\leq\frac{P-p}{M-m}$; then the function is convex; and
the `cup' version:
\begin{equation}
A_{1b}\left(  t\right)  :=\min\left[  m+\left(  t-p\right)  /d,\;M\right]
\text{ ,} \label{basprocup}%
\end{equation}
where $\frac{p}{m}\leq d\leq\frac{P-p}{M-m}$; in this case the function is
concave. Note that not only the choice of the parameter $d$, but also the
choice of one of two forms of the base+prop function ($A_{1a}$ or $A_{1b}$)
depends on other constraints ($p_{i}, i=1,\ldots,N$, and $S$) in
(\ref{main}), see also Sect.~\ref{ALL3}.
Observe further that the base+prop+floor and base+prop+cup functions are in a sense
extremal allocation functions satisfying boundary conditions: $A\left(
p\right)  =m$ and $A\left(  P\right)  =M$, since it is clear that every such
function must fulfill the inequalities:
\begin{equation}
\max\left[  m,\;\left(  M/P\right)  t\right]  \leq A(t)\leq\min\left[  \left(
m/p\right)  t,\;M\right]  \label{basprobou}%
\end{equation}
for $t\in\left[  p,P\right]  $, and thus it is bounded from below by a
base+prop+floor function ($d:=P/M$), and from above by a base+prop+cup
function ($d:=p/m$).

\item \textsl{piecewise linear} functions:
\begin{equation}
A_{2a}(t):=\max\left[  m+\left(  t-p\right)  /d,\;\left(  M/P\right)
t\right]  \text{ ,} \label{pielinconvex}%
\end{equation}

where $\frac{P-p}{M-m}\leq d$; the function is convex; or
\begin{equation}
A_{2b}(t):=\min\left[  \left(  m/p\right)  t,\;\left(  t-P\right)
/d+M\right]  \text{ ,} \label{pielinconcave}%
\end{equation}
where $\frac{P-p}{M-m}\leq d$; the function is concave. Again, the choice of
one of two forms of the piecewise linear function ($A_{2a}$ or $A_{2b}$)
depends on constraints in (\ref{main}).

\item \textsl{quadratic} (\textsl{parabolic}) functions:
\begin{equation}
A_{3}\left(  t\right)  :=\left(  \frac{t-p}{P-p}\frac{M}{P}+\frac{P-t}%
{P-p}\frac{m}{p}\right)  t-d\left(  t-p\right)  \left(  P-t\right)  \text{ .}
\label{par}%
\end{equation}

Depending on the system constraints $M,m,P,p$ and the parameter $d$ determined
by the total size $S$ of the House, the function is convex or concave. In
particular, if
\begin{equation}
0\leq d-\Theta\leq\frac{\min\left(  M-m,m-Mp/P\right)  }{\left(  P-p\right)
^{2}} \label{par1} \text{ ,}%
\end{equation}
with $\Theta:=\frac{m/p-M/P}{P-p}$, the function (\ref{par}) is convex. In the
case
\begin{equation}
0\geq d-\Theta\geq-\frac{\min\left(  M-m,mP/p-M\right)  }{\left(  P-p\right)
^{2}} \label{par2}%
\end{equation}
the parabolic allocation function is concave.

\item[4.] \textsl{base+power} functions:
\begin{equation}
A_{4}(t):=M\frac{t^{d}-p^{d}}{P^{d}-p^{d}}+m\frac{P^{d}-t^{d}}{P^{d}-p^{d}%
}\text{ ,} \label{baspow}%
\end{equation}
where either $0<d\leq1$ and $\left(  M/m-1\right)  d\leq\left(  P/p\right)
^{d}-1$, or $1<d$ and $\left(  1-m/M\right)  d\leq1-\left(  p/P\right)  ^{d}$.
In the first case the function is concave, in the second convex. In the
limiting case ($d\rightarrow0$) we get a \textsl{logarithmic} function:
\begin{equation}
A_{l}(t):=\frac{\ln\left(  P^{m}/p^{M}\right)  +\left(  M-m\right)  \ln t}%
{\ln\left(  P/p\right)  }\text{ ,} \label{log}%
\end{equation}
which is an allocation function, if $M/m-1\leq\ln\left(  P/p\right)  $.

\item[5.] \textsl{homographic} functions:
\begin{equation}
A_{5}\left(  t\right)  :=\frac{M\left(  t/M-d\right)  \left(  t-p\right)
+m\left(  t/m-d\right)  (P-t)}{\left(  P/M-d\right)  \left(  t-p\right)
+\left(  p/m-d\right)  \left(  P-t\right)  }\text{ ,} \label{hom}%
\end{equation}
where either $d\leq p/M$ or $d\geq P/m$. In the first case the function is
concave, in the second convex. In the limiting case ($d\rightarrow\pm\infty$)
we get a linear function.
\end{enumerate}

All five families discussed above share a common element: the \textsl{linear}
(\textsl{affine}, $\frac{d}{dt}A_{\operatorname*{lin}}\left(  t\right)
\equiv\operatorname*{const}\geq0$) function $A_{\operatorname*{lin}}:\left[
p,P\right]  \rightarrow\left[  m,M\right]  $ given by the formula
\begin{equation}
A_{\operatorname*{lin}}(t):=M\frac{t-p}{P-p}+m\frac{P-t}{P-p}\text{ .}
\label{lin}%
\end{equation}
On the other hand, if $\frac{d}{dt}\frac{A\left(  t\right)  }{t}\equiv
const\leq0$, then $A$ must be a quadratic function given by (\ref{par}) with
$d=0$, i.e.,
\begin{equation}
A_{q}\left(  t\right)  :=\left(  \frac{t-p}{P-p}\frac{M}{P}+\frac{P-t}%
{P-p}\frac{m}{p}\right)  t\text{ .} \label{spepar}%
\end{equation}

Some of the above solutions were already discussed in the literature, also in
the context of the European Parliament.

The base+prop class which seems to lead to the simplest of all these methods
was first analysed in \citet{Puk07,Puk10a}, see also
\citet{MarRam08,MarRam10}, and became the basis for the recent
proposal, called \textsl{`Cambridge Compromise'}, elaborated in January 2011,
and discussed later by the Committee on Constitutional Affairs (AFCO) of the
European Parliament \citep{Gri11,Grietal11a}. Here we present this method in the
`spline' form, see \citet{MarRam08}. Likewise, one of the methods of
apportionment of seats in the projected Parliamentary Assembly of the United
Nations is based on this model \citep[p. 25]{Bum10}. Note that, in fact,
the composition of the Electoral College that formally elects the President
and Vice President of the United States of America also reflects the base+prop
scheme, where each state is allocated as many electors as it has Senators
(equal base) and Representatives (proportional representation, with at least one
seat per state) in the United States Congress. The idea of combining these two
approaches to the apportionment problem was first put forward by one of the
Founding Fathers of the United States and the future American President,
James Madison in 1788 \citep{Mad1788}.

The quadratic (parabolic) method was proposed and advocated by Ram\'{i}rez
Gonz\'{a}lez and his co-workers in a series of papers
\citep{Ram04,Rametal06,MarRam08,MarRam10,Ram10}.

The methods of apportionment of seats in the European Parliament using
base+power functions were also considered by several authors, see
\citet{TheSch77,Rametal06,Arn08,MarRam08,MarRam10,SloZyc10,Grietal11b}
and \citet{Mob11}. Note that a similar method was proposed for solving the taxation problem
already in the nineteenth century by a Dutch economist Arnold Jacob
Cohen-Stuart \citep{Coh1889}. Moreover, the variant of this method (using the
square-root function) was also considered in \citet[p.~27]{Bum10}.

As far as we know, out of five families presented above, only the piecewise
linear family has not yet been analysed in detail in the European Parliament
context, since the homographic functions have been independently studied
under the name of projective quotas by \citet{Ser11}. On the other hand,
yet another class of `linear-hyperbolic' functions was used both in the
apportionment problem for the European Parliament \citep{SloZyc10} as well as
in the tax schedule proposed by a Swedish economist Karl Gustav Cassel at the
beginning of the twentieth century \citep{Cas1901}. Note, that also the proportional
apportionment method with minimum and maximum requirements \citep[p.~133]{BalYou01,MarRam08} can be
described within this general framework, taking (neither concave nor convex) apportionment
function $A$ given by $A(t)=\operatorname*{med}(m,dt,M)$, where $M/P<d<m/p$, and
$\operatorname*{med}$ stands for the median value of three.

For a simple and general algorithm of constructing families of allocation
functions see Sect.~\ref{DEG}.

\section{Allocation functions - necessary and sufficient
conditions\label{ALL2}}

In this section we present several simple propositions that give necessary and
sufficient conditions for a function $A:\left[  p,P\right]  \rightarrow\left[
m,M\right]  $ to be a (degressive) allocation function. Almost all these facts
belong to mathematical folklore, but we provide short proofs here for the
completeness of presentation. First of all, observe that an allocation
function needs to be continuous, because, as a non-decreasing function, it can
only have jump discontinuities, but this contradicts degressive proportionality.

We start from a simple characterization of allocation functions.

\begin{proposition}
\label{Peeinepro}$A$ is an allocation function if and only if
\begin{equation}
\frac{A\left(  s\right)  }{A\left(  t\right)  }\leq\max\left(  1,\frac{s}%
{t}\right)  \text{ ,} \label{Peeine}%
\end{equation}
or equivalently
\begin{equation}
\min\left(  1,\frac{s}{t}\right)  \leq\frac{A\left(  s\right)  }{A\left(
t\right)  } \label{Peeine2}%
\end{equation}
for every $s,t\in\left[  p,P\right]  $.
\end{proposition}

See also \citet[p. 327]{Pee70}.

\begin{proof}
Let $s<t$, then (\ref{Peeine}) is equivalent to $A\left(  s\right)  /A\left(
t\right)  \leq1$. On the other hand for $s>t$ we get $A\left(  s\right)
/A\left(  t\right)  \leq s/t$, as desired.
\end{proof}

Note that $A$ need not be neither concave nor convex. (Consider, e.g., the
allocation function $A:\left[ 2,8\right]  \rightarrow [ \sqrt
{2}+1/2,2\sqrt{2}+1/8 ]  $ given by $A(t)=\sqrt{t}+1/t$ for $2\leq
t\leq8$, that has an inflection point at $t=4$.) However, if $A$ is an
allocation function, then it can be bounded from above by its greatest 
convex minorant and from below by its least concave majorant. Because
of this, it cannot be neither `too convex' nor `too concave'.

\begin{corollary}
If $A$ is an allocation function, then
\begin{align}
\frac{1+\sqrt{p/P}}{2}\overline{A}\left(  t\right)   &  \leq\frac{t\left(
P-p\right)  }{P\left(  t-p\right)  +t\left(  P-t\right)  }\overline{A}\left(
t\right)  \leq\label{bounds}\\
A\left(  t\right)   &  \leq\frac{t\left(  P-p\right)  }{p\left(  P-t\right)
+t\left(  t-p\right)  }\underline{A}\left(  t\right)  \leq\frac{1+\sqrt{P/p}%
}{2}\underline{A}\left(  t\right)  \text{ ,}\nonumber
\end{align}
for each $t\in\left[  p,P\right]  $, where $\underline{A}$ and $\overline{A}$
denote, respectively, the greatest convex minorant function and the least
concave majorant of $A$ (i.e. the largest convex function smaller than $A$ and
the smallest concave function larger than $A$).
\end{corollary}

\begin{proof}
For $t\in\left[  p,P\right]  $ we have $\overline{A}\left(  t\right)
=\sup\sum_{i=1}^{n}\lambda_{i}A\left(  t_{i}\right)  $, where the sum runs
over $\lambda_{i}\geq0$, $t_{i}\in\left[  p,P\right]  $, $n\in\mathbb{N}$,
satisfying $\sum_{i=1}^{n}\lambda_{i}=1$, $\sum_{i=1}^{n}\lambda_{i}t_{i}=t$.
From (\ref{Peeine}) we get $\sum_{i=1}^{n}\lambda_{i}A\left(  t_{i}\right)
\leq A\left(  t\right)  \sum_{i=1}^{n}\lambda_{i}\max\left(  1,\frac{t_{i}}%
{t}\right)  \leq\frac{2Pt-Pp-t^{2}}{t\left(  P-p\right)  }A\left(  t\right)
$. Hence $A\left(  t\right)  \geq\frac{t\left(  P-p\right)  }{2Pt-Pp-t^{2}%
}\overline{A}\left(  t\right)  \geq\frac{1+\sqrt{p/P}}{2}\overline{A}\left(
t\right)  $. The proof for the greatest convex minorant is analogous.
\end{proof}

The next proposition gives a sufficient condition for a convex or concave
non-decreasing function to be an allocation function.

\begin{proposition}
If $A$ is non-decreasing, concave and fulfills $A\left(  t\right)  /t\leq
A\left(  p\right)  /p$ for all $t\in\left[  p,P\right]  $, or if it is
non-decreasing, convex and satisfies $A\left(  t\right)  /t\geq A\left(
P\right)  /P$ for all $t\in\left[  p,P\right]  $, then $A$ is an allocation
function. In particular, every concave function $A:\left[  0,+\infty\right)
\rightarrow\left[  0,+\infty\right)  $ is an allocation function restricted to
any interval $\left[  p,P\right]  $ for $0<p<P$.
\end{proposition}

\begin{proof}
In the former case to show that $A$ is degressively proportional, it is enough
to observe that $A\left(  s\right)  /s=A\left(  \frac{t-s}{t-p}\cdot
p+\frac{s-p}{t-p}\cdot t\right)  /s\geq A\left(  p\right)  \frac{t-s}{\left(
t-p\right)  s}+A\left(  t\right)  \frac{s-p}{\left(  t-p\right)  s}\geq
A\left(  t\right)  /t$ for $s,t\in\left[  p,P\right]  $, $s<t$, as required.
The proof for convex functions is analogous.
\end{proof}

In fact, if $A:\left[  p,P\right]  \rightarrow\left[  m,M\right]  $ is a
restriction of the function defined on the interval $\left[  0,P\right]  $
such that $A\left(  0\right)  =0$, then, to get degressive proportionality, it
is enough to assume that $A$ is concave on average, i.e., that the function
$\left[  0,P\right]  \ni t\rightarrow a\left(  t\right)  :=\frac{1}{t}\int
_{0}^{t}A\left(  s\right)  ds\in\left[  0,M\right]  $ is concave, since
$A(t)/t=a^{\prime}\left(  t\right)  +a\left(  t\right)  /t$ for $0<t\leq P$
and both components are non-increasing functions of $t$ in this case, see
\citet[Theorem 5]{BruOst62}.

We call a function $A:\left[  p,P\right]  \rightarrow\left[  m,M\right]  $
\textsl{subadditive} if $A\left(  s+t\right)  \leq A\left(  s\right)
+A\left(  t\right)  $ holds for every $s,t,s+t\in\left[  p,P\right]  $. The
subadditivity is the necessary condition for a function being an allocation
function, as the next proposition shows. (Analogously, in taxation
progressivity of income tax implies its merging-proofness, see \citet[Corollary
1]{JuMor11}.)

\begin{proposition}
\label{sub}If $A$ is an allocation function, then $A$ is subadditive.
\end{proposition}

See also \citet[Theorem 1.4.3]{Ros50} and \citet[Theorem 7.2.4]{HilPhi57}.

\begin{proof}
Let $s,t,s+t\in\left[  p,P\right]  $. From the degressive proportionality we get 
$A\left(  s+t\right)  /\left(s+t\right)  \leq\min\left(  A\left(  s\right)  /s,A\left(  t\right)
/t\right)  $. Hence $A\left(  s+t\right)  \leq \left(s+t\right) 
\left( \frac{s}{s+t}\frac{A\left(
s\right)  }{s}+\frac{t}{s+t}\frac{A\left(  t\right)  }{t}\right)=A\left(  s\right)
+A\left(  t\right)  $.
\end{proof}

The converse implication fails in general, but it holds for convex and
non-decreasing functions.

\begin{corollary}
If $A:\left[  p,+\infty\right)  \rightarrow\left[  m.+\infty\right)  $ is
convex and non-decreasing, then $A$ is an allocation function if and only if
it is subadditive.
\end{corollary}

See \citet[Theorem 1.4.6]{Ros50}.

\begin{proof}
According to Proposition \ref{sub} it is enough to show that convex,
non-decreasing and subadditive function is degressively proportional. Let
$p<s<t$. Then $A\left(  t\right)  \leq\frac{s}{t}A\left(  s\right)  +\left(
1-\frac{s}{t}\right)  A\left(  s+t\right)  \leq\frac{s}{t}A\left(  s\right)
+\left(  1-\frac{s}{t}\right)  \left(  A\left(  s\right)  +A\left(  t\right)
\right)  =A\left(  s\right)  +\left(  1-\frac{s}{t}\right)  A\left(  t\right)
$. Hence $A\left(  t\right)  /t\leq A\left(  s\right)  /s$, as desired.
\end{proof}

\section{Allocation functions - concave or convex?\label{ALL3}}

Analyzing possible schemes of allocating seats in the European Parliament
several authors consider only concave allocation functions
\citep{MarRam08,MarRam10}. However, as we have seen above, in the class of
degressively proportional functions convex and concave functions seem to play
similar roles, and both types of functions are represented in each of five
basic classes considered.

The affine allocation function (which lies on the border between the concave
and the convex realm) can serve as a solution of the apportionment problem if
and only if $\sum_{i=1}^{N}A_{\operatorname*{lin}}\left(  p_{i}\right)
\approx S$. This, however, is only an approximate statement because the effect
is influenced by the rounding procedure. Thus, in a concrete case, whether
convex or concave functions should be used in the allocation scheme depends
approximately on the sign of the expression $\sum_{i=1}^{N}%
A_{\operatorname*{lin}}\left(  p_{i}\right)  -S$. Taking into account that
\begin{equation}
\sum\nolimits_{i=1}^{N}A_{\operatorname*{lin}}\left(  p_{i}\right)
-S=  \frac{\left(  \left\langle P\right\rangle -p\right)
\left(  \left\langle M\right\rangle -m\right)  N}{P-p}  \left(
\mu-\rho\right)  \text{ ,}\label{calculation}%
\end{equation}
with%
\begin{equation}
\rho:=\frac{P-\left\langle P\right\rangle }{\left\langle P\right\rangle
-p}\label{rho}%
\end{equation}
and
\begin{equation}
\mu:=\frac{M-\left\langle M\right\rangle }{\left\langle M\right\rangle
-m}\text{ ,}\label{mu}%
\end{equation}
where $\left\langle P\right\rangle $ and $\left\langle M\right\rangle $
denote, respectively, the mean population of a country and the mean number of
seats per country, we see that the solution of the dilemma depends on which of
two numbers is greater $\rho$ or $\mu$. If $\rho\geq\mu$ one should use
concave functions for resolving the problem, if $\rho\leq\mu$, convex. Since%
\begin{gather*}
\left(  \mu-\rho\right)  \left(  \left\langle P\right\rangle -p\right)
\left(  \left\langle M\right\rangle -m\right)  N^{2}=\\
T\left(  M-m\right)  -\left(  S\left(  P-p\right)  -N\left(  mP-Mp\right)
\right)  \text{ ,}%
\end{gather*}
where $T$ is the total population of the Union, $S$ is the size of the House,
and $N$ denotes the number of the Member States, the inequality $\rho\geq\mu$
can be rewritten in the following form affine in $T$, $S$, and $N$:
\begin{equation}
T\leq S\cdot\frac{P-p}{M-m}-N\cdot\frac{mP-Mp}{M-m}\text{ .}\label{TSN}%
\end{equation}
In particular, this implies that any accession of a new state of moderate size
(to leave $p$ and $P$ unchanged) to the Union (which means $T,N\uparrow$),
keeping `constitutional' parameters ($m$, $M$, $S$) fixed, reduces the
probability of finding concave solution of the apportionment problem.
Furthermore, the right hand side of (\ref{TSN}) is a decreasing function of
both $m$ and $M$ (as long as $Nm<S<NM$, which is both a natural and necessary
assumption) and an increasing function of $S$. In consequence, seeking concave
solutions, one has either to enlarge the size of the House, or to lower the
number of seats assigned to the smallest or to the largest Member State (or both).

Note, however, that the treaties define only the minimal ($M_{\min}$) and
maximal ($M_{\max}$) numbers of seats in the Parliament, requiring merely that
$m:=A\left(  p\right)  \geq M_{\min}$ and $M:=A\left(  P\right)  \leq M_{\max
}$, as well as the value of $S$. While we have to set the exact values of $m$
and $M$ to start the allocation procedure described in Sect.~\ref{INT}, our
choice is formally limited only by these inequalities. Thus, if we believe
that the concavity is a desirable feature of an allocation function and it
should be possibly incorporated to its definition, we have to agree that the
enlargement process will result at some point (defined in fact by the equality
in (\ref{TSN})) in lowering the value of $M$ below $M_{\max}$. The only other
solution of this problem one can imagine is to introduce an amendment to the
treaty either decreasing the minimal number of seats $M_{\min}$ or increasing
the total number of seats $S$. However, these two alternatives may be
difficult to accept for political reasons, and in this case decreasing the
number $M$ seems to be the most feasible solution of the problem within the
`concave' realm.

\section{Degressive proportionality through logarithmic eyes\label{DEG}}

We believe that it is sometimes better to analyse allocation functions in
logarithmic (log-log)\ coordinates, since this approach provides us with a
number of benefits, namely:

\begin{itemize}
\item  It is more convenient to plot a graph of population-seats relationship
in these coordinates, and so, to compare different allocation methods, since
we have more small than large member states in the European Union. NB, this is
quite a natural situation from the statistical point of view (`the larger the fewer').

\item  In this setting it is easier to express our assumptions (monotonicity
and degressive proportionality) in a uniform way.

\item  This approach gives us a better framework to analyse certain additional
properties of allocation methods.
\end{itemize}

\begin{definition}
Define $L:\left[  \ln p,\ln P\right]  \rightarrow\left[  \ln
m,\ln M\right]  $ by
\begin{equation}
L(\ln t):=\ln A(t)
\end{equation}
for $x\in\left[  \ln p,\ln P\right]  $. In other words, $L = \ln \circ \, A \circ \exp$
or $A=\exp \circ \, L \circ \ln$.
\end{definition}

The choice of a logarithmic base corresponds to the choice of a unit and is
not important here.

\begin{proposition}
Assume that a function $A:\left[  p,P\right]  \rightarrow\left[
m,M\right]  $ is differentiable. Then the following equivalences are true:\medskip%

\begin{tabular}
[c]{ccc}%
$A$ is non-decreasing & $\Leftrightarrow$ & $L^{\prime}\geq0$\medskip\\
$A$ is degressively proportional & $\Leftrightarrow$ & $L^{\prime}\leq
1$\medskip\\
$A$ is an allocation function & $\Leftrightarrow$ & $0\leq L^{\prime}\leq
1$\medskip
\end{tabular}
.
\end{proposition}

In particular, the above statement gives us a clear mathematical
interpretation of degressive proportionality. Now, our task can be reduced to
a search for a function $L:\left[  \log p,\log P\right]  \rightarrow\left[
\log m,\log M\right]  $ fulfilling $0\leq L^{\prime}\leq1$. These can be
smoothly realized in a three-fold way:

\begin{enumerate}
\item $L$ is affine (i.e. $L^{\prime}$ is constant, i.e. $L^{\prime}\equiv
c\in\left[  0,1\right]  $);

\item $L$ is convex (i.e. $L^{\prime}$ increases from, say, $0$ to $1$) (i.e.
$A$ is \textsl{geometrically convex}, see \citet{Mat97});

\item $L$ is concave (i.e. $L^{\prime}$ decreases from, say, $1$ to $0$) (i.e.
$A$ is \textsl{geometrically concave}, see \citet{Mat97}).
\end{enumerate}

The first scenario leads to the \textsl{power function }(or in other words, a
base+power function with the base $0$) given by $A(t):=b \, t^{d_{\ast}}$, where
\begin{equation}
d_{\ast}:=\left(  \ln\left(  M/m\right)  \right)  /\left(  \ln\left(  P/p\right)
\right)  \label{exponent}%
\end{equation}
and
\begin{equation}
b := \left(M-m\right)/(P^{d}-p^{d}) = e^{[  \left(  \ln
m\right)  \left(  \ln P\right)  -\left(  \ln M\right)  \left(  \ln p\right)]
/\ln\left(  P/p\right)  }\text{ .} \label{base}%
\end{equation}

Rather surprisingly, the distinction between the second and third possibility
seems to have a clear interpretation in terms of properties of allocation
function $A$, namely,\textsl{\ }the properties of \textsl{sub-} and
\textsl{superproportionality}. The notion of subproportionality and the dual
notion of superproportionality were introduced into the decision theory by
Daniel Kahneman, a Nobel Prize laureate in economy, and Amos Tversky, a
mathematical psychologists, in 1979 \citep{KahTve79} and since then used by
many authors, see e.g. \citet{AlDha10}. Let us recall their definition.

\begin{definition}
We say that $A$ is \textit{superproportional} (\textit{subproportional}) iff
for every $s,t\in\operatorname*{dom}(A)$, $s\leq t$ and $0\leq r\leq1$ such
that $rs,rt\in\operatorname*{dom}(A)$ we have
\begin{equation}
\frac{A(rs)}{A(rt)}\geq\frac{A(s)}{A(t)}\text{ \ \ }\left(  \frac
{A(rs)}{A(rt)}\leq\frac{A(s)}{A(t)}\right)  \text{ .} \label{supsub}%
\end{equation}
\end{definition}

\begin{proposition}
Let $L:\left[  \log p,\log P\right]  \rightarrow\left[  \log m,\log M\right]
$ and $A=\exp\circ L\circ\log$. The following equivalences hold:

\begin{itemize}
\item $L$ is convex iff $A$ is superproportional;

\item $L$ is concave iff $A$ is subproportional.
\end{itemize}
\end{proposition}

\begin{proof}
Note that $A$ is superproportional iff $L(b+a)-L(b)\leq L(c+a)-L(c)$ for $\log
p\leq b\leq c\leq c+a\leq\log P$. This property is equivalent to convexity of
$L$. The proof of the second equivalence is analogous.
\end{proof}

To illustrate this property consider two pairs of member states,
Romania/France and Lithuania/Hungary, with the similar population
quotient ($s/t\approx1/3$) and another such configuration: Finland/Portugal
and Latvia/Ireland ($s/t\approx1/2$). In Tab.~\ref{tab1}. the values of seat quotients
for five methods analysed in Sect. \ref{ALL1} are shown. Note that in all these cases the seat
quotient for the `smaller' pair is greater than for the `larger' one.

\begin{table}[htbp]
\centering
\caption{Population ratio (PQ) for exemplary pairs of member states and the
corresponding quotients of the number of seats (SQ) in the European Parliament
for five classes of allocation functions: $1$ = \textit{base+prop},
$2$ = \textit{piecewise linear}, $3$ = \textit{parabolic},
$4$ = \textit{base+power}, $5$ = \textit{homographic} with
 the rounding to
the nearest integer.}\smallskip
\begin{tabular}{ccccccc}
\toprule
{\small ratio} & {\small PQ} & {\small SQ1} & {\small SQ2}& {\small SQ3}& {\small SQ4}& {\small SQ5} \\
\midrule
{\small RO/FR} & {\small 0.332} & {\small 0.376} & {\small 0.397} &
{\small 0.413} & {\small 0.418} & {\small 0.413}\\
{\small LT/HU} & {\small 0.332} & {\small 0.556} & {\small 0.632} &
{\small 0.526} & {\small 0.526} & {\small 0.526}\\
{\small Fl/PT} & {\small 0.503} & {\small 0.677} & {\small 0.737} &
{\small 0.684} & {\small 0.650} & {\small 0.684}\\
{\small LV/IE} & {\small 0.503} & {\small 0.727} & {\small 0.769} &
{\small 0.727} & {\small 0.750} & {\small 0.727}\\
\bottomrule
\end{tabular}%
\label{tab1}%
\end{table}%

Using other words, a \textsl{superproportional method} leads to the following
property of an allocation system (at least before rounding):

\begin{center}
\textit{The smaller a pair of states is, the larger is the gain}

\textit{of the small member in the pair over the large one.}
\end{center}

Thus, this is in fact a kind of \textsl{degressive-degressive proportionality}%
. It is easy to show that if an allocation function $A$ is subproportional,
then it must be concave, and if it is convex it is necessarily superproportional.

This approach leads also to a simple algorithm for constructing allocation
functions, see also \citet[Sect. 4]{AlDha10}. Choose a continuous function
$h:\left[  p,P\right]  \rightarrow\left[  0,1\right]  $ such that
\begin{equation}
\int_{p}^{P}\frac{h(s)}{s}ds=\ln\left(  M/m\right)  \text{ .} \label{intcon}%
\end{equation}
Solving the first-order homogeneous linear differential equation of the form
\begin{equation}
A^{\prime}\left(  x\right)  =\frac{h(x)}{x}A\left(  x\right)  \label{diffeq}%
\end{equation}
with the initial condition $A\left(  p\right)  =m$ we get the allocation
function given by the formula
\begin{equation}
A(t)=m\exp\left(  \int_{p}^{t}\frac{h(s)}{s}ds\right)  \label{allfor}%
\end{equation}
that fulfills also the final condition $A(P)=M$. In fact, every differentiable
allocation function can be obtain in this way. Moreover, $A$ is
superproportional (resp. subproportional) iff $h$ is increasing (resp.
decreasing), which provides a simple test for checking superproportionality.

To illustrate this technique consider the function $h:\left[  p,P\right]
\rightarrow\left[  0,1\right]  $ given by
\begin{equation}
h\left(  t\right)  =\frac{d}{1+ct^{-d}}\text{ ,} \label{hbaspow}%
\end{equation}
where the exact value of $c:=\frac{mP^{d}-Mp^{d}}{M-m}$ is determined by the
integral condition (\ref{intcon}), and we assume additionally that either
$d_{1}<d\leq1$, where
\begin{equation}
d_{1}:=\inf\{0<d<1:\left(  M/m-1\right)  d\leq\left(  P/p\right)
^{d}-1\}\text{ ,} \label{d1}%
\end{equation}
or $1<d<d_{2}$, where
\begin{equation}
d_{2}:=\sup\{d>1:\left(  1-m/M\right)  d\leq1-\left(  p/P\right)  ^{d}\}\text{
,} \label{d2}%
\end{equation}
in order to ensure that $0\leq h\leq1$. Applying (\ref{allfor}) we get a
base+power function $A$ given by (\ref{baspow}).

Clearly, the function $h$ defined by (\ref{hbaspow}) is increasing for $c>0$
and decreasing for $c<0$, and so the necessary and sufficient condition for
$A$ being superproportional (resp. subproportional) in this case is that $c>0$
(resp. $c<0$) or equivalently $d>d_{\ast}$ (resp. $d<d_{\ast}$), where
$d_{\ast}$ is given by (\ref{exponent}) and $d_{1}<d_{\ast}<1$.

Summarizing, we have five possible forms of the base+power allocation function:

\begin{itemize}
\item  concave and subproportional function for $d_{1}<d<d_{\ast}$;

\item  power function for $d=d_{\ast}$;

\item  concave and superproportional function for $d_{\ast}<d<1$;

\item  affine function for $d=1$;

\item  convex and superproportional function for $1<d<d_{2}$.
\end{itemize}

Note, however, that in a concrete situation the choice of the value of $d$ is determined by the
constraint (\ref{main}).

\section{The European Parliament\label{EP}}

For the European Parliament we have the following values of parameters:
$p=412\,970$, $M_{\min}=6$, $P=81\,802\,257$, $M_{\max}=96$, $T=501\,103\,425$%
, $S=751$, and $N=27$. Assuming that the upper and the lower bounds are
saturated, $m=M_{\mathrm{min}}$ and $M=M_{\mathrm{max}}$ we obtain
$\rho\approx3.\,485\geq3.\,126\approx\mu$, so our choice of an allocation
function is limited to concave functions. However, it follows from (\ref{TSN})
that for the Parliament of size $703$ or less we would have to find the
solution of the apportionment problem in the realm of convex functions or
otherwise to relax the constraints considering some $M<M_{\max}$. (Due to
rounding, this number may be somewhat smaller, cf. \citet{Kel11}.) This means also that,
in fact, we have currently only about fifty seats to allocate freely besides
the linear (or, saying more precisely, affine) distribution.

Analyzing five families of allocation functions and three rounding methods we
get fifteen possible solutions for the apportionment problem, see Tab.~\ref{tab2}.

\begin{sidewaystable*}[htbp]
\caption{Fifteen solutions of the apportionment problem for the
European Parliament (five classes of allocation functions: $1$ =
\textit{base+prop}, $2$ = \textit{piecewise linear}, $3$ = \textit{parabolic},
$4$ = \textit{base+power}, $5$ = \textit{homographic}; three rounding methods:
$d$ = \textit{downwards}, $m$ = \textit{to the nearest integer}, $u$ =
\textit{upwards}); LT = the distribution of seats under the Lisbon Treaty;
population figures are taken from the Eurostat website (OJ 22.12.2010 L 338/47).}
\smallskip
\centering
\begin{tabular}{cccccccccccccccccc}
    \toprule
    Country & Population & LT    & 1d    & 1m    & 1u    & 2d    & 2m    & 2u    & 3d    & 3m    & 3u    & 4d    & 4m    & 4u    & 5d    & 5m    & 5u \\
    \midrule
    Germany & 81802257 & 96    & 96    & 96    & 96    & 96    & 96    & 96    & 96    & 96    & 96    & 96    & 96    & 96    & 96    & 96    & 96 \\
    France & 64714074 & 74    & 86    & 85    & 83    & 77    & 78    & 78    & 81    & 80    & 80    & 79    & 79    & 79    & 80    & 80    & 80 \\
    United Kingdom & 62008048 & 73    & 82    & 81    & 80    & 74    & 75    & 75    & 78    & 78    & 77    & 76    & 76    & 76    & 77    & 77    & 77 \\
    Italy & 60340328 & 73    & 80    & 79    & 78    & 73    & 73    & 73    & 76    & 76    & 75    & 74    & 74    & 74    & 76    & 76    & 75 \\
    Spain & 45989016 & 54    & 62    & 62    & 61    & 57    & 57    & 58    & 62    & 61    & 60    & 60    & 59    & 59    & 61    & 61    & 60 \\
    Poland & 38167329 & 51    & 53    & 52    & 51    & 49    & 49    & 49    & 53    & 52    & 52    & 52    & 51    & 51    & 53    & 52    & 51 \\
    Romania & 21462186 & 33    & 32    & 32    & 32    & 31    & 31    & 31    & 33    & 33    & 32    & 33    & 33    & 32    & 33    & 33    & 32 \\
    Netherlands & 16574989 & 26    & 26    & 26    & 26    & 26    & 26    & 26    & 27    & 27    & 26    & 27    & 27    & 27    & 27    & 27    & 26 \\
    Greece & 11305118 & 22    & 19    & 19    & 19    & 20    & 20    & 20    & 20    & 20    & 20    & 21    & 21    & 20    & 20    & 20    & 20 \\
    Belgium & 10839905 & 22    & 19    & 19    & 19    & 20    & 20    & 19    & 20    & 20    & 20    & 20    & 20    & 20    & 20    & 20    & 20 \\
    Portugal & 10637713 & 22    & 18    & 18    & 19    & 19    & 19    & 19    & 19    & 19    & 19    & 20    & 20    & 20    & 20    & 19    & 19 \\
    Czech Republic & 10506813 & 22    & 18    & 18    & 18    & 19    & 19    & 19    & 19    & 19    & 19    & 20    & 20    & 19    & 19    & 19    & 19 \\
    Hungary & 10014324 & 22    & 17    & 18    & 18    & 19    & 19    & 19    & 19    & 19    & 19    & 19    & 19    & 19    & 19    & 19    & 19 \\
    Sweden & 9340682 & 20    & 17    & 17    & 17    & 18    & 18    & 18    & 18    & 18    & 18    & 18    & 18    & 18    & 18    & 18    & 18 \\
    Austria & 8375290 & 19    & 15    & 16    & 16    & 17    & 17    & 17    & 16    & 16    & 16    & 17    & 17    & 17    & 17    & 17    & 17 \\
    Bulgaria & 7563710 & 18    & 14    & 15    & 15    & 16    & 16    & 16    & 15    & 15    & 15    & 16    & 16    & 16    & 15    & 15    & 15 \\
    Denmark & 5534738 & 13    & 12    & 12    & 13    & 14    & 14    & 14    & 13    & 13    & 13    & 13    & 13    & 13    & 13    & 13    & 13 \\
    Slovakia & 5424925 & 13    & 12    & 12    & 12    & 14    & 14    & 14    & 12    & 13    & 13    & 13    & 13    & 13    & 13    & 13    & 13 \\
    Finland & 5351427 & 13    & 12    & 12    & 12    & 14    & 14    & 14    & 12    & 13    & 13    & 13    & 13    & 13    & 12    & 13    & 13 \\
    Ireland & 4467854 & 12    & 11    & 11    & 11    & 13    & 13    & 13    & 11    & 11    & 12    & 12    & 12    & 12    & 11    & 11    & 12 \\
    Lithuania & 3329039 & 12    & 9     & 10    & 10    & 12    & 12    & 11    & 10    & 10    & 10    & 10    & 10    & 11    & 10    & 10    & 10 \\
    Latvia & 2248374 & 9     & 8     & 8     & 9     & 11    & 10    & 10    & 8     & 8     & 9     & 9     & 9     & 9     & 8     & 8     & 9 \\
    Slovenia & 2046976 & 8     & 8     & 8     & 8     & 10    & 10    & 10    & 8     & 8     & 9     & 8     & 9     & 9     & 8     & 8     & 9 \\
    Estonia & 1340127 & 6     & 7     & 7     & 8     & 10    & 9     & 9     & 7     & 7     & 8     & 7     & 7     & 8     & 7     & 7     & 8 \\
    Cyprus & 803147 & 6     & 6     & 6     & 7     & 9     & 9     & 9     & 6     & 7     & 7     & 6     & 7     & 7     & 6     & 7     & 7 \\
    Luxembourg & 502066 & 6     & 6     & 6     & 7     & 7     & 7     & 8     & 6     & 6     & 7     & 6     & 6     & 7     & 6     & 6     & 7 \\
    Malta & 412970 & 6     & 6     & 6     & 6     & 6     & 6     & 6     & 6     & 6     & 6     & 6     & 6     & 6     & 6     & 6     & 6 \\

    EU-27 & 501103425 & 751   & 751   & 751   & 751   & 751   & 751   & 751   & 751   & 751   & 751   & 751   & 751   & 751   & 751   & 751   & 751 \\ \bottomrule

\end{tabular}%
\label{tab2}%
\end{sidewaystable*}

Observe that all these solutions are quite similar, which is a consequence of
the fact that our choice is limited by two factors: the predetermined shape of
the graph of an allocation function, and the fact that more than ninety
percent of seats are in a sense distributed in advance. More precisely, the
results for the parabolic, base+power, and homographic allocation functions
are almost identical, whereas the choice of the base+prop functions is
advantageous for large countries, and the choice of the piecewise linear
functions seems to be beneficial for small countries.

The influence of the choice of a rounding method on the distribution of seats
is a non-trivial mathematical problem even for proportional apportionment \citep{BalYou01,Jan11},
where it was proven that, statistically, the rounding downwards is more often
advantageous for large countries and the rounding upwards for small countries,
see \citet{Schetal03,DrtSch05,Sch08}. In the case of the European Parliament
one can observe a similar effect for the base+linear, parabolic, base+power
and homographic functions, where the rounding downwards is the best
possibility and the rounding upwards is the worst for large countries (from
the Netherlands to France), whereas for small countries (from Malta to
Austria) the situation is reversed. However, for the piecewise linear class we
find completely different pattern, and so it is not clear to what extent this
rule applies to degressively proportional apportionment.

\begin{figure*}
\includegraphics[width=0.5\textwidth]{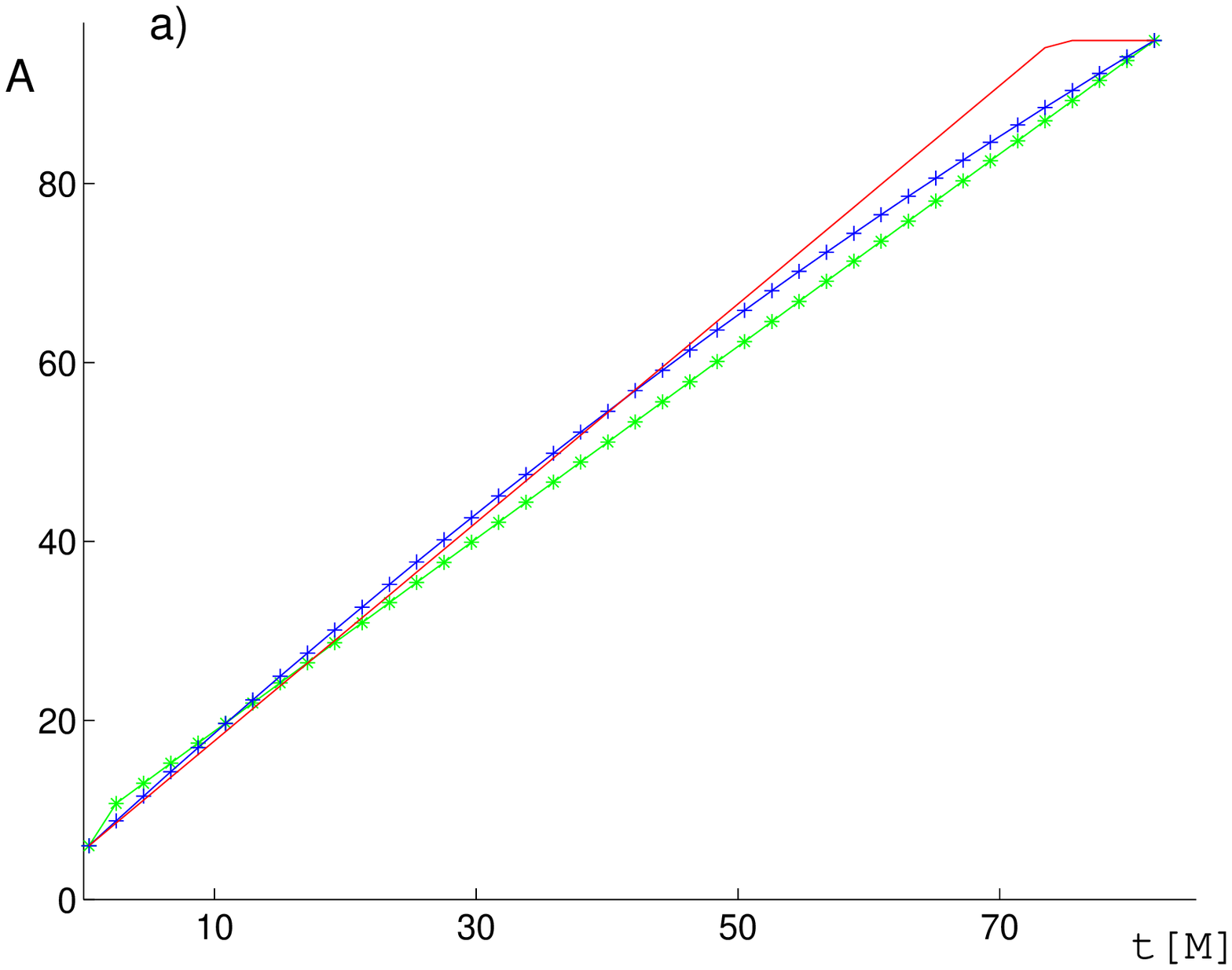}
\includegraphics[width=0.5\textwidth]{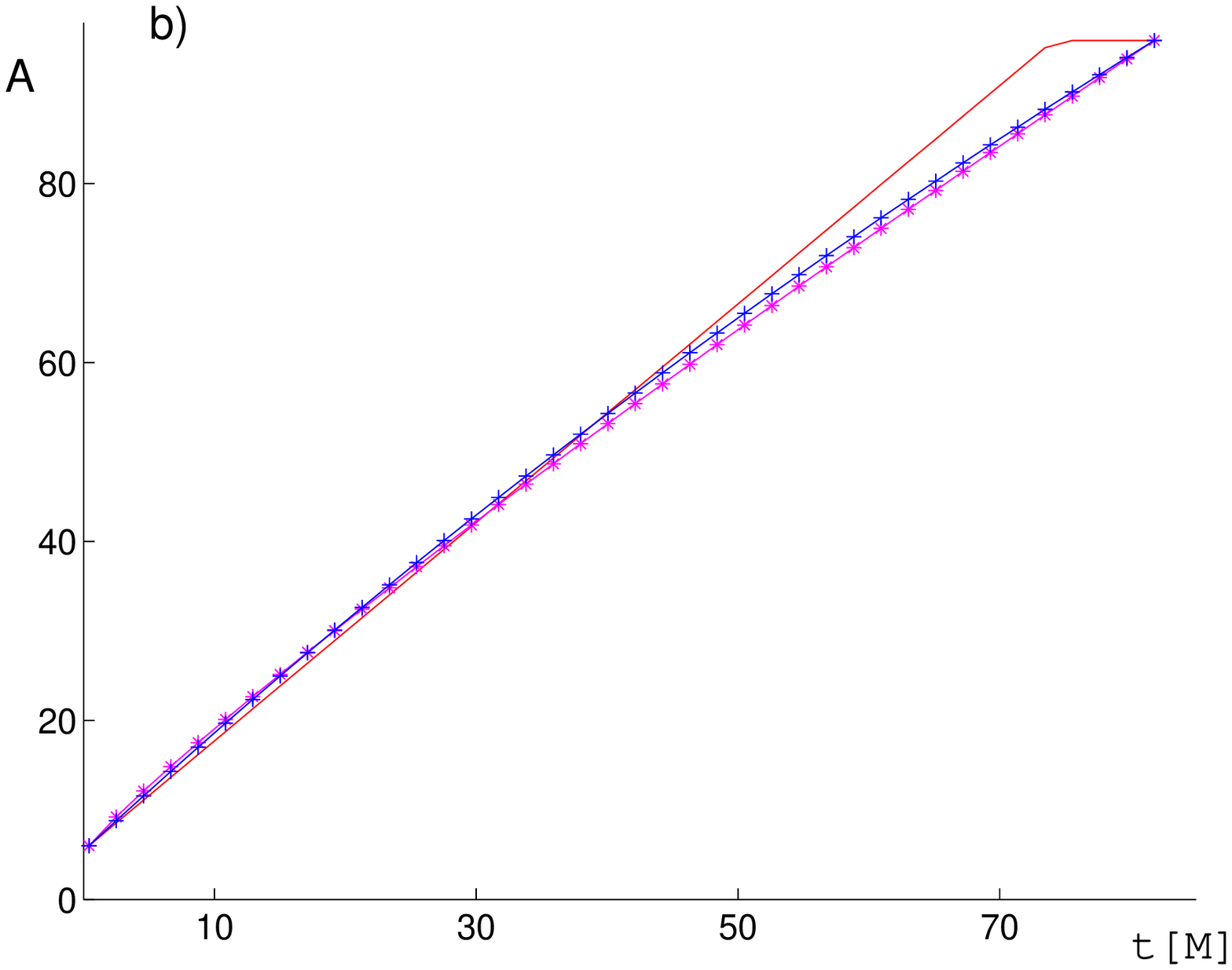}
\caption{Allocation functions applied to the European
Parliament: panel a) concave \textit{base+prop} function (\ref{basprocup}) (solid
line), \textit{piecewise linear} function (\ref{pielinconcave}) ($\ast$), and \textit{parabolic}
function (\ref{par}) ($+$); panel b) function (\ref{basprocup}) drawn as a
reference solid line, \textit{base+power} function (\ref{baspow}) ($\ast$), and
\textit{homographic} function (\ref{hom}) ($+$). The argument $t$ denotes the
population of a state in millions, while $A$ is scaled to determine the
corresponding number of seats in the Parliament consisting of $S=751$ members
with the constraints $m=6$ and $M=96$ seats. }%
\label{fig1}
\end{figure*}

As regards superproportionality, the base+prop method is superproportional in
the `affine' part of its domain, i.e. for all countries but the largest one,
the piecewise linear method for all countries but two smallest ones, and the
parabolic (resp. homographic) method are superproportional for small and
medium countries and subproportional for large five (resp. six) ones.

The only one of the five methods that is superproportional in the whole domain
$\left[  p,P\right]  $\ is the base+power method. In fact, we showed that this
method is superproportional as long as $d>d_{\ast}$, where $d_{\ast}$ is given
by (\ref{exponent}). In the analysed case $d_{\ast}\approx0.\,524$ and $d=0.865$,
$0.894$, $0.922$ depending on the rounding method chosen, so the condition is
clearly fulfilled. Though it is not known whether superproportionality is what
the authors of the Lisbon Treaty really intended, when they formulated the
`degressive proportionality' rule, we think that it is worth to realize that
the base+power method fulfills it for all pairs, whereas the other methods can
violate it for some countries. Thus base+power method is in a sense more
degressively proportional, or one can say degressively proportional in more
perfect way, than other methods analysed above. Incidentally, the base+power
solution with $c=0.5$ (the square root) results (with downward rounding) in a
round number of $1000$ members of the Parliament.

In \citet{Grietal11a} the authors decided to select the method called
`Cambridge Compromise', which is in this case equivalent to the base+prop
method (as defined above) with the rounding to the nearest integer, mainly
because of its obvious simplicity. However, this solution has been criticized
for being `not enough degressively proportional' \citep{Mob11} and departing
too much from the \textit{status quo}. In \citet{Grietal11b} the solution very
similar to the base+power method discussed here is considered `as a step along
a continuous transition from the negotiated status quo composition to the
constitutionally principled Cambridge Compromise.' (Indeed this method is
closest to the \textit{status quo} out of all methods analyzed in Tab.~\ref{tab2}.) The
crucial point in these discussions seems to be the meaning of the term
`degressive proportionality'. Is it only a lame form of (pure)
proportionality, as it was actually suggested in \citet{Grietal11b}
or is it a separate notion that requires distinct mathematical and political
solutions, as \citet{Mob11} claims? In this paper we have tried to shed new
light on this debate, analyzing mathematical properties of degressively
proportional allocation functions and indicating the differences between
various classes of such functions.

If we are looking for a degressively proportional (resp. degressively
proportional and superproportional) and increasing function, in the log-log
realm we have to find a function (resp. convex function) with the derivative
contained between $0$\ and $1$. Adding to this, three constraints related to
the minimum and maximum number of seats and to the size of the House, we see
that our choice is in fact very limited and all the solution satisfying these
conditions must look quite similar -- see Fig. 2.

The key possibility to vary the allocation schemes considerably is to change
the number $M$ of the seats allotted to the largest member state. As specified
in the Treaty of Lisbon the upper bound reads $M_{\mathrm{\max}}=96 $, but
this bound needs not to be saturated and one may also take $M<M_{\mathrm{\max
}}$. By doing so, one introduces more freedom into the space of possible
solutions, as more seats can be allotted besides the affine distribution.

Note also that by extending the Union and keeping the number $M$ fixed (which
is, however, in the `concave realm', doable only up to a certain total
population of the Union), the seats for the new member states are donated by
all but the largest state. If any further enlargement of the Union was
performed according to this scheme, the ratio of the seats in the European
Parliament allocated to the largest state would remain constant. In
consequence, as the number $N$ of the member states was increased, the voting
power of the largest state in the European Union would grow.

These arguments show that the choice of the number $M$ selected to design an
allocation system is crucial. The issue: under what conditions the constraint
$M=M_{\mathrm{\max}}$ should be relaxed seems to be equally important as the
choice of the actual form of allocation function. As regards the latter, it is
rather difficult task to distinguish in practice one of them. From an academic
perspective, however, it would be interesting to base the solution of the
`degressive' allocation problem on an axiomatic approach, possibly considering
some additional properties of allocation functions as concavity and superproportionality.

\medskip

\textit{Acknowledgements.} It is a pleasure to thank Geoffrey Grimmett and
Friedrich Pukelsheim for inviting us to the Cambridge Apportionment Meeting,
where this work was initiated, and all the participants of this meeting for
fruitful and stimulating discussions, as well as Axel Moberg for providing us
with an earlier version of his paper and for interesting correspondence.
\end{document}